\documentclass[10pt]{article}
\usepackage{amsthm,amsmath,amssymb,upgreek}
\usepackage{subfig,sidecap,caption}
\usepackage{paralist}
\usepackage[usenames,dvipsnames]{color}
\usepackage[pdftex,breaklinks,colorlinks,
    citecolor={BlueViolet}, linkcolor={Blue},urlcolor=Maroon]{hyperref}  
\usepackage{tikz,pgfkeys}
\usepackage{tkz-graph}
\usetikzlibrary{arrows,shapes}
\usetikzlibrary{decorations.pathreplacing, decorations.shapes}
\usepackage{graphicx}
\usepackage{charter,eulervm}%
\usepackage{multirow,booktabs,array}
\usepackage{tabularx}
\usepackage[final,expansion=alltext,protrusion=true]{microtype}

\theoremstyle{plain} %remark,  
\newtheorem{theorem}{Theorem}[section]
\newtheorem{lemma}[theorem]{Lemma}

\newtheorem{proposition}[theorem]{Proposition}
\newtheorem{reduction}{Reduction}

\def\boxit#1{\vbox{\hrule\hbox{\vrule\kern4pt
      \vbox{\kern1pt#1\kern1pt} \kern2pt\vrule}\hrule}}

\newcommand{\opt}[1]{\ensuremath{{\varphi(#1)}}}
\newcommand{\ptas}{polynomial time approximation scheme}

\title{Minimum Fill-In: Inapproximability and \\Almost Tight Lower
  Bounds\thanks{Supported in part by the Hong Kong Research Grants
    Council (RGC) under grant 252026/15E and the National
    Natural Science Foundation of China (NSFC) under grants 61572414
    and 61420106009.}}

\author{Yixin Cao\thanks{Department of Computing, Hong Kong
    Polytechnic University, Hong Kong, China.
    \href{mailto:yixin.cao@polyu.edu.hk}{\tt yixin.cao@polyu.edu.hk}.}
  \and R. B. Sandeep\thanks{Department of Computer Science \&
    Engineering, Indian Institute of Technology Hyderabad, India.
    \href{mailto:cs12p0001@iith.ac.in}{\tt cs12p0001@iith.ac.in}.  The
    bulk of this work was done while visiting Hong Kong
    Polytechnic University.}  }

\date{}

\begin{document}
\maketitle
\begin{abstract}
  Given an $n \times n$ sparse symmetric matrix with $m$ nonzero
  entries, performing Gaussian elimination may turn some zeroes into
  nonzero values.  To maintain the matrix sparse, we would like to
  minimize the number $k$ of these changes, hence called the minimum
  fill-in problem.
  Agrawal, Klein, and Ravi [FOCS 1990;
  \href{http://dx.doi.org/10.1007/978-1-4613-8369-7_2} {Graph Theory
    \& Sparse Matrix Comp.\ 1993}] developed the first approximation
  algorithm for the problem, based on early work on heuristics by
  George [\href{http://dx.doi.org/10.1137/0710032} {{SIAM J. Numer.\
      Anal.\ 10(1973)}}] and by Lipton, Rose, and Tarjan
  [\href{http://dx.doi.org/10.1137/0716027} {SIAM J. Numer.\ Anal.\
    16(1979)}].  The objective function they used is $m + k$, the
  number of nonzero elements in the matrix after elimination.  An
  approximation algorithm using $k$ as the objective function was
  presented by Natanzon, Shamir, and Sharan [STOC 1998;
  \href{http://dx.doi.org/10.1137/S0097539798336073}{SIAM J. Comput.\
    30(2000)}].  These two versions are incomparable to each other in
  terms of approximation.

  Parameterized algorithms for the problem was first studied by
  Kaplan, Shamir, and Tarjan [FOCS 1994;
  \href{http://dx.doi.org/10.1137/S0097539796303044}{SIAM J. Comput.\
    28(1999)}].  Fomin and Villanger [SODA 2012; \href
  {http://dx.doi.org/10.1137/11085390X}{SIAM J. Comput.\ 42(2013)}]
  recently gave an algorithm running in time $2^{O(\sqrt{k} \log k)} +
  n^{O(1)}$.

  Hardness results of this problem are surprisingly scarce, and the
  few known ones either are weak or have to use nonstandard complexity
  conjectures.  The only inapproximability result by Wu et
  al.~[IJCAI'15; \href{http://dx.doi.org/10.1613/jair.4030}{J.\
    Artif.\ Intell.\ Res.\ 49(2014)}] applies to only the objective
  function $m + k$, and is grounded on the Small Set Expansion
  Conjecture.  The only nontrivial parameterized lower bounds, by
  Bliznets et
  al.~[\href{http://dx.doi.org/10.1137/1.9781611974331.ch79}{SODA
    2016}], include a very weak one based on the Exponential Time
  Hypothesis (\textsc{eth}), and a strong one based on hardness of
  subexponential-time approximation of the minimum bisection problem
  on regular graphs.  For both versions of the problem, we exclude the
  existence of polynomial time approximation schemes, assuming
  P$\ne$NP, and the existence of $2^{O(n^{1-\delta})}$-time
  approximation schemes for any positive $\delta$, assuming
  \textsc{eth}.  It also implies a $2^{O(k^{1/2 - \delta})} \cdot
  n^{O(1)}$ parameterized lower bound.  Behind these results is a new
  reduction from vertex cover, which might be of its own interest: All
  previous reductions for similar problems are from some kind of
  graph layout problems.
\end{abstract}

\thispagestyle{empty}
\setcounter{page}{0}
\newpage
\section{Introduction}\label{sec:intro}

The minimum fill-in problem arises from the application of Gaussian
elimination to a sparse matrix, where we want to minimize the number
of zero entries that are turned into nonzero values, called the
\emph{fill-in}, during the elimination process.  As usual, here we
ignore the accidental transformation from a nonzero value to zero.
Rose~\cite{rose-72-sparse-matrix} gave a graph-theoretic
interpretation of the minimum fill-in problem on symmetric matrices.
A graph $G$ can be easily extracted from an $n\times n$ symmetric
matrix $M$ as follows: introduce $n$ vertices, each for a row, and an
edge between the $i$th vertex and the $j$th vertex if and only if $i
\ne j$ and $M_{i j} \ne 0$.  The minimum fill-in of $M$ is exactly the
minimum number of edges we need to add to $G$ to make it chordal,
i.e., free of induced cycles of length of at least four.

This correlation turns out to be crucial in the sense that almost all
known algorithmic and hardness results to this problem use the graph
formulation.  One early heuristic approach is to choose a vertex with
the minimum degree and add edges to make its neighborhood a clique
\cite{george-81-sparse-positive-definite}, behind which is the
observation that every chordal graph has a vertex whose neighborhood
form a clique.  Using another fact that every minimal separator of a
chordal graph is a clique, George \cite{george-73-nested-dissection,
  george-81-sparse-positive-definite} proposed the \emph{nested
  dissection} heuristic, which recursively finds a balanced separator
and add edges to make it a clique.  Its performance relies thus on how
good the separators we can find; e.g., combined with the Lipton-Tarjan
planar separator theorem \cite{lipton-79-planar-separator}, it
immediately leads to approximation algorithms for planar and
bounded-genus graphs \cite{lipton-79-generalized-nested-dissection,
  gilbert-87-nested-dissection}.
Using the approximation algorithm of Leighton and
Rao~\cite{leighton-99-multicommodity-in-approximation} for finding
balanced separators, Agrawal et al.~\cite{agrawal-93-minimum-fill-in}
developed the first algorithm with a nontrivial performance guarantee
for the minimum fill-in problem on general
graphs.

The minimum fill-in problem, minimizing the number of added edges (the
number of zeroes turned into nonzero values), can also be formulated
as minimizing the number of edges in the resulting chordal supergraph
(the number of nonzero elements in the matrix after elimination).
When discussing approximation algorithms, we need to specify the
objective functions: They may behave quite different from the point of
view of approximation.  For convenience, we use the name \emph{minimum
  fill-in} when the objective function is the size of the fill-in, and
\emph{chordal completion} otherwise.

Let $G$ be a graph on $n$ vertices and $m$ edges, and let $\opt{G}$
denote the size of minimum fill-ins of $G$.  The algorithm of Agrawal
et al.~\cite{agrawal-93-minimum-fill-in, agrawal-91-thesis} always
produces a chordal supergraph of at most $O\big((m +
\opt{G})^{0.75}\sqrt{m}\log^{3.5}n \big)$ edges, thereby having a
ratio $O(\sqrt[4]{m}\log^{3.5}n) = O(\sqrt{n}\log^{3.5}n)$ for the
chordal completion problem.  The first (and only) approximation
algorithm for minimum fill-in was reported by Natanzon et
al.~\cite{natanzon-00-approximate-fill-in}, which has a ratio $8
\opt{G}$, i.e., it always finds a fill-in of size at most $8
\varphi^2(G)$.  We remark that these two results are incomparable in
general.
They also provided algorithms with better approximation ratios on
graphs of degrees at most $d$, $O( \sqrt{d}\log^4n)$ for chordal
completion \cite{agrawal-93-minimum-fill-in} and $O(d^{2.5}\log^4(n
d))$ \cite{natanzon-00-approximate-fill-in} for minimum fill-in.

Thus far there are no constant-ratio approximation algorithms known
for either of them.  Even more embarrassing might be the progress on
hardness results.  We could not even exclude \ptas{s} for the minimum
fill-in problem.  The only known inapproximability result for chordal
completion was given by Wu et
al.~\cite{wu-14-inapproximability--treewidth}, who excluded
constant-ratio approximation on the assumption of the Small Set
Expansion Conjecture, which is related to the Unique Games Conjecture
but less established \cite{raghavendra-10-graph-expansion-and-ugc}.
We give the first inapproximability result for both problems on the
assumption P $\ne$ NP.
\begin{theorem}\label{thm:no-ptas}
  If either of the minimum fill-in and the chordal completion problems
  has a \ptas{}, then P $=$ NP.
\end{theorem}

We actually show a stronger result, which however needs a stronger
complexity assumption, namely the Exponential Time Hypothesis, which
states that the satisfiability problem with at most 3 variables per
clause (\textsc{3sat}) cannot be solved in $2^{o(m + n)}$ time, where
$m$ and $n$ denote the number of clauses and variables in the Boolean
formula~\cite{impagliazzo-01-eth, impagliazzo-01-seth}.  The
Exponential Time Hypothesis (\textsc{eth}) is the standard working
hypothesis of fine-grained complexity, which aims to understand the
exact time complexity of problems and to prove lower
bounds.\footnote{For the reader unfamiliar with this line of research,
  we refer to \cite{lokshtanov-11-lower-bound-on-ETH} and
  ~\cite[Chapter~14]{cygan-15} for reference.}
\begin{theorem}\label{thm:main-2}
  Assuming \textsc{eth}, there is some positive $\epsilon$ such that
  no algorithm can find a ($1 + \epsilon$) approximation for the
  minimum fill-in problem or the chordal completion problem in time
  $2^{O(n^{1 - \delta})}$, for any positive constant $\delta$.
  % or in time $2^{O(k^{1/2 - \delta})} \cdot n^{O(1)}$.
\end{theorem}

This makes also a significant contribution to an important problem in
parameterized computation, i.e., the parameterized lower bound of
minimum fill-in.  Recall that given a graph $G$, and an integer
parameter $k$, the parameterized minimum fill-in problem asks whether
there is a fill-in of size at most $k$.  (Note that it does not make
much sense to use $m + k$ as the parameter.)  Kaplan et
al.~\cite{kaplan-99-chordal-completion} designed the first
parameterized algorithm for the minimum fill-in problem, which runs in
time $O(16^k k^6 + k^2 mn)$, and proposed an $O(k^3)$-vertex kernel.
Natanzon et al.~\cite{natanzon-00-approximate-fill-in} manged to
improve it to $O(k^2)$, which played a crucial role in their
approximation algorithm mentioned above.
Fomin and Villanger~\cite{fomin-12-subexponential-fill-in} developed a
parameterized algorithm running in time $2^{O(\sqrt{k} \log k)} +
O(k^2 mn)$, thereby placing this problem in the class of very few
problems that admit subexponential-time parameterized algorithms on
general graphs.  

Note that the problem is trivial when $k > n^2/2$; otherwise $k^{1/2}
= O(n)$.  Thus, Theorem~\ref{thm:main-2} immediately implies an almost
tight lower bound on parameterized algorithms for this problem.

\begin{theorem}\label{thm:main-3}
  Assuming \textsc{eth}, there is no algorithm that can solve the
  minimum fill-in problem in time $2^{O(k^{1/2 - \delta})} \cdot
  n^{O(1)}$, for any positive constant $\delta$.
\end{theorem}

Let us put Theorem~\ref{thm:main-3} into context.  Again, compare to
the algorithmic progress, the hardness results of parameterized
algorithms lay far behind.  The only nontrivial lower bounds were
given by Bliznets et al.~\cite{bliznets-16-lower-bounds} very
recently.
To relate their results, however, we need to start from 1970s.  The
complexity of the minimum fill-in problem was among the open problems
of Garey and Johnson \cite{garey-79}, and settled by
Yannakakis~\cite{yannakakis-81-minimum-fill-in} with a simple
reduction from the optimal linear arrangement problem.  His reduction,
however, is not very much helpful for deriving inapproximability and
other hardness results we want for the problem.  For example, there is
no inapproximability result for the optimal linear arrangement problem
on the assumption P $\ne$ NP; to exclude polynomial time approximation
schemes for it, Amb{\"{u}}hl~\cite{ambuhl-11-inapproximability} had to
use the assumption that NP-complete problems cannot be solved in
randomized subexponential time.

On the other hand, to derive lower bounds on exact or parameterized
algorithms, we need to prevent the graph or the parameter of the
reduced instance from increasing too much with the reduction.
Therefore, if we want to reuse the reduction of Yannakakis, we need to
trace the whole sequence of reductions from 3\textsc{sat}, which, if
we spell out, has five steps, namely, max-2\textsc{sat}, maximum cut,
optimal linear arrangement, and chain completion, before eventually
minimum fill-in.  As said, the last two reductions are by
Yannakakis~\cite{yannakakis-81-minimum-fill-in}; while the first three
reductions are due to Garey et
al.~\cite{garey-76-simplified-problems}.  For the prospect of deriving
tight lower bounds on the minimum fill-in problem from Yannakakis'
reduction, the main obstacles lie in step~3, from maximum cut to
optimal linear arrangement, and step~4, from optimal linear
arrangement to chain completion.  In
\cite{garey-76-simplified-problems}, the original version of the
reduction for step~3 blows up an $n$-vertex graph to an $n^4$-vertex
graph.  The selection of $4$ in the exponent turns out to be for the
convenience of presentation, and any constant larger than $3$ would
suffice.  Step~4 then blows up the graph size by another quadratic
factor.  Therefore, as already mentioned in
\cite{fomin-12-subexponential-fill-in, bliznets-16-lower-bounds}
(without explanation), assuming \textsc{eth}, one can only derive
lower bounds of $2^{o(\sqrt[3]{n})}$ and $2^{o(\sqrt[6]{k})}\cdot
n^{O(1)}$ from these reductions in their original form
\cite{garey-76-simplified-problems, yannakakis-81-minimum-fill-in}.

Bliznets et al.~\cite{bliznets-16-lower-bounds} sedulously retraced
the tortuous (and torturous) five-step reduction, and managed to
decrease the size of the reduced instance by a stronger complexity
conjecture and heavier constructions.  First, they managed to improve
step~3 such that it produces a graph of almost linear size, which
allows them to derive improved but still weaker lower bounds of
$2^{O(\sqrt{n}/\log^c n)}$ and $2^{O(\sqrt[4]{k}/\log^c k)}\cdot
n^{O(1)}$ for some constant $c$.  Then, to avoid the quadratic
explosion in step~4, they had to introduce a new conjecture on the
subexponential-time approximation hardness of the minimum bisection
problem on $d$-regular graphs.

Instead of further improving these reductions or finding alternatives,
we start from scratch and devise a completely new reduction, which
turns out to be surprisingly simple.  Our reduction is from the vertex
cover problem, whose NP-hardness is derived directly from \textsc{sat}
\cite{karp-72-reduction}.  The following theorem summarizes a simple
version of our reduction, which serves an alternative, and far
simpler, proof for the NP-hardness of the minimum fill-in problem.
\begin{theorem}\label{thm:np-reduction}
  Given an $n$-vertex graph $G$, we can construct in polynomial time
  another graph $H$ on $n^3 + n$ vertices such that $G$ has a vertex
  cover of size $c$ if and only if $H$ has a fill-in of size at most
  $(c + 1) n^2 - 1$.
\end{theorem}

The rest of the paper is organized as follows.
Section~\ref{sec:simple-reduction} presents the reduction summarized
by Theorem~\ref{thm:np-reduction}.  Section~\ref{sec:no-ptas} improves
it to a linear reduction that proves
Theorem~\ref{thm:no-ptas}--\ref{thm:main-3}: It uses instead a graph
of bounded degree, and a gap version of vertex cover.
Section~\ref{sec:remarks} contrasts our reduction to that of
Yannakakis, and proposes some possible remedy to make our reduction
work for the related interval completion problem.

\section{The simple reduction}\label{sec:simple-reduction}

In the rest of the paper, we exclusively use the graph formulation of
the minimum fill-in problem and hence graph-theoretic notion.  More
details on this formulation can be found in
\cite{rose-72-sparse-matrix}.  All graphs discussed in this paper are
undirected and simple.  The vertex set and edge set of a graph $G$ are
denoted by $V(G)$ and $E(G)$ respectively.
For $\ell\ge 4$, we use $C_\ell$ to denote a hole on $\ell$ vertices.
Chordal graphs are precisely $\{C_\ell:\ell\ge 4\}$-free graphs.  A
graph is a \emph{split graph} if its vertices can be partitioned into
a clique and an independent set.  It is known that a graph is a {split
  graph} if and only if it contains no $2 K_2$ (a $4$-vertex graph
with two edges that share no vertex), $C_4$, or $C_5$ as an induced
subgraph.  Since every $C_\ell$ with $\ell \ge 6$ contains a $2 K_2$,
a split graph is necessarily chordal.

A set of vertices is a \emph{vertex cover} of a graph $G$ if every
edge of $G$ has at least one end in this set.  A \emph{fill-in} of a
graph $G$ is a set $E_+$ of non-edges of $G$ such that $G + E_+$
(denoting the graph with vertex set $V(G)$ and edge set $E(G)\cup
E_+$) is chordal.  We use $\tau(G)$ to denote the size of minimum
vertex covers of the graph $G$, and $\opt{G}$ to denote the size of
minimum fill-ins of the graph $G$.  The size of chordal completion $G
+ E_+$ is then $|E(G)| + |E_+|$, and the size of a minimum chordal
completion is then $|E(G)| + \opt{G}$.  From the aspect of
optimization, there is no difference between minimizing $|E_+|$, i.e.,
the number of added edges, and $|E(G)| + |E_+|$, i.e., the number of
edges in the resulting graph.  They behave, however, quite different
with respect to approximation; consider, for example, a dense graph
($|E(G)| = \Omega(n^2)$) .  In approximation algorithms for minimum
fill-in, the ratio is defined to be $|E_+|/\opt{G}$, while for chordal
completion it is $( |E(G)| + |E_+| )/ ( |E(G)| + \opt{G})$.

We give here the reduction announced in
Theorem~\ref{thm:np-reduction}.

\begin{reduction}\label{reduction:premitive}
  Let $G$ be a graph on $n$ vertices.  For each vertex $v$ of $G$,
  introduce a set of $n^2$ new vertices, and add edges to make them
  adjacent to all other vertices of $G$ but $v$ itself.  Let $U$
  denote the set of $n^3$ new vertices.  Add all possible edges to
  connect $U$ into a clique.
\end{reduction}

\begin{figure}[t!]
  \centering
    \includegraphics{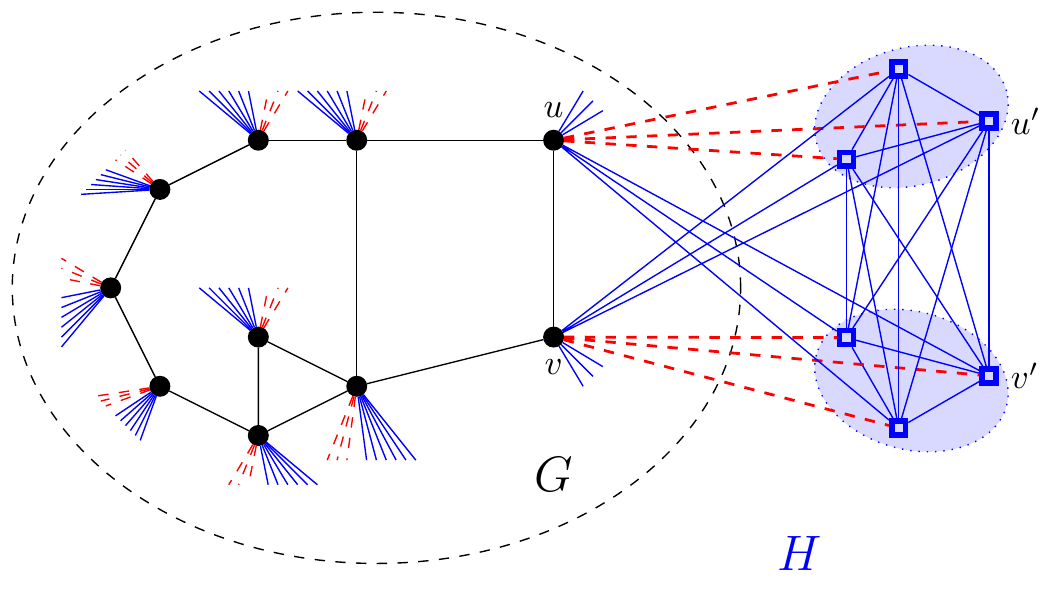}   
\caption{Illustration for Reduction~\ref{reduction:premitive}.  The
  original graph $G$, with black vertices and black edges, is inside
  the dashed ellipse.  Only two sets of the new vertices (blue
  squares) are shown in the right of the figure; they are
  corresponding to $u, v\in V(G)$ respectively.  The $n^2$ vertices
  for $u$ are connected to $v$ by blue edges, and they are nonadjacent
  to $u$, indicated by red dashed lines; likewise for the $n^2$ new
  vertices for $v$.  Edges between them and other vertices of $G$ are
  omitted for clarity.  The whole graph is $H$.}
  \label{fig:np-reduction}
\end{figure}

An example of Reduction~\ref{reduction:premitive} is illustrated in
Figure~\ref{fig:np-reduction}.  Let $H$ be the graph obtained from $G$
by Reduction~\ref{reduction:premitive}.  It has $n + n \cdot n^2 =
O(n^3)$ vertices.  Apart from the edges from $E(G)$, the graph $H$
contains $n^3 (n^3 - 1) / 2$ edges among $U$, and $n \cdot (n^3 -
n^2)$ edges between $V(G)$ and $U$, hence $E(H) = O(n^6)$.
For any given vertex cover $C$ of $G$, the set $V(G)\setminus C$ is an
independent set of $G$; since we do not add edges in $V(G)$ during the
reduction, $V(G)\setminus C$ is also an independent set of $H$.
Therefore, if we add all edges among $C\cup U$ to make it a clique,
then we end with a split graph.  Recall that a split graph is chordal;
we have thus constructed a fill-in of $H$ from a vertex cover of $G$.

We now consider the other direction, i.e., how to extract a vertex
cover of $G$ from a fill-in of $H$.  We say that a vertex $v$ of $G$
is \emph{full} with respect to a fill-in $E_+$ if $E_+$ contains all
the missing edges between $v$ and $U$, and we simply say it is full
when the fill-in $E_+$ is clear from the context.  The following
simple observation is crucial for the whole paper.

\begin{proposition}\label{lem:lower-bound}
  Let $H$ be the graph obtained from $G$ by
  Reduction~\ref{reduction:premitive}.  For any fill-in $E_+$ of $H$,
  the set $C$ of vertices that are full with respect to $E_+$ is a
  vertex cover of $G$.
\end{proposition}
\begin{proof}
  Let $\widehat H = H + E_+$ and let $u v$ be any edge of $G$.  We
  argue that at least one of $u$ and $v$ is full with respect to
  $E_+$.  Suppose otherwise, then we can find $u', v'\in U$ such that
  $u u'$ and $v v'$ are non-edges of $\widehat H$; see
  Figure~\ref{fig:np-reduction}.  By the construction, $u'$ and $v'$
  are distinct and $u v', u' v, u' v'\in E(H)\subseteq E(\widehat H)$,
  and then $u v u' v'$ is a $C_4$ of $\widehat H$, contradicting that
  $\widehat H$ is chordal.  Therefore, between the two vertices of
  each edge of $G$, at least one of them is full with respect to
  $E_+$, and the proposition follows.
\end{proof}

We are now ready to present our main result of this section, which is
an easy consequence of the aforementioned two-way constructions.

\begin{lemma}\label{thm:simple-reduction}
  Let $G$ be an $n$-vertex graph, and let $H$ be the graph obtained
  from $G$ by Reduction~\ref{reduction:premitive}.  Then $\tau(G) n^2
  \le \opt{H} < (\tau(G) + 1) n^2$.
\end{lemma}
\begin{proof}
  The lower bound follows directly from
  Proposition~\ref{lem:lower-bound}.  For the upper bound, we
  construct a fill-in for $H$ as follows.  Let $C$ be any minimum
  vertex cover of $G$.  We add all edges among $C\cup U$ to make $G$ a
  split graph.  These include $|C| n^2$ edges between $C$ and $U$ and
  ${|C| \choose 2} - |E(G[C])|$ non-edges in $G[C]$. Thus,
  \begin{align*}
    \opt{H} \le |C| n^2 + {|C| \choose 2} - |E(G[C])|
    \le |C| n^2 + {|C| \choose 2}
    = \tau(G) n^2 + { \tau(G) \choose 2}
    < \tau(G) n^2 + {n \choose 2}
    < \big(\tau(G) + 1\big) n^2.
  \end{align*}
  This concludes the proof.
\end{proof}

To see why Lemma~\ref{thm:simple-reduction} implies
Theorem~\ref{thm:np-reduction}, note that $\text{ if } \tau(G) \le
c$, then $\opt{H} \le \big(\tau(G) + 1\big) n^2 - 1 \le (c+1)n^2 - 1$; otherwise, 
$\opt{H} \ge \tau(G) n^2 \ge (c+1)n^2 > (c+1)n^2 - 1$.

\section{The hardness results}\label{sec:no-ptas}
A graph is \emph{$d$-degree-bounded} if every vertex has degree
at most $d$.
Any {$d$-degree-bounded} graph can be trivially colored with $d+1$
colors.  According to Brooks' theorem \cite{brooks-41}, if $d \ge 3$
and the graph does not contain a clique on $d + 1$ vertices, then it
can be colored by $d$ colors, and such a coloring can be found in
linear time.  Our second reduction starts from such a colored
{$d$-degree-bounded} graph.  We introduce a set of vertices for each color
class instead of each vertex, and the size of each set is $b n$ for
some constant $b$ to be specified later.

\begin{reduction}\label{reduction:colored}
  Let $G$ be an $n$-vertex {$d$-degree-bounded} graph that does not contain a
  clique on $d + 1$ vertices.  We find a proper coloring of $G$ with
  $d$ colors.  For each color, introduce a set of $b n$ new vertices,
  and add edges to make them adjacent to all vertices of $G$ that
  receive a different color.  Let $U$ denote all the $b d n$ new
  vertices; add all possible edges to connect $U$ into a clique.
\end{reduction}

An example of Reduction~\ref{reduction:colored} is illustrated in
Figure~\ref{fig:fine-reduction}.  The produced graph $H$ has $(b d +
1) n = O(n)$ vertices.
For any vertex cover $C$ of $G$, adding all edges among $C\cup U$ to
$H$ makes it a split graph.  Thus,
\begin{equation}\label{eq:upper-bond}
  \opt{H} \le \tau(G) \cdot b n + {\tau(G) \choose 2} < b n \tau(G) +
  \frac{1}{2} \tau^2(G).
\end{equation}
Although we use a constant number of sets, the following facts from
Reduction~\ref{reduction:premitive} remain true: (1) each vertex of
$G$ is nonadjacent to one set of new vertices; and (2) for each edge
of $G$, its two vertices are nonadjacent to two different sets of new
vertices.  The second fact is ensured by the proper coloring.  We
define the full vertices in exactly the same way as before: a vertex
$v$ of $G$ is \emph{full} with respect to a fill-in $E_+$ if $E_+$
contains all the missing edges between $v$ and $U$.  It is easy to
verify that Proposition~\ref{lem:lower-bound} remains true for the new
reduction: Actually, a word-by-word copy of the proof works.
\begin{proposition}\label{lem:vertex-cover}
  Let $G$ be an $n$-vertex {$d$-degree-bounded} graph that does not contain a
  clique on $d + 1$ vertices, and let $H$ be the graph obtained from
  $G$ by Reduction~\ref{reduction:colored}.  For any fill-in $E_+$ of
  $H$, the set of vertices that are full with respect to $E_+$ is a
  vertex cover of $G$.
\end{proposition}

\begin{figure}[t]
  \centering
    \includegraphics{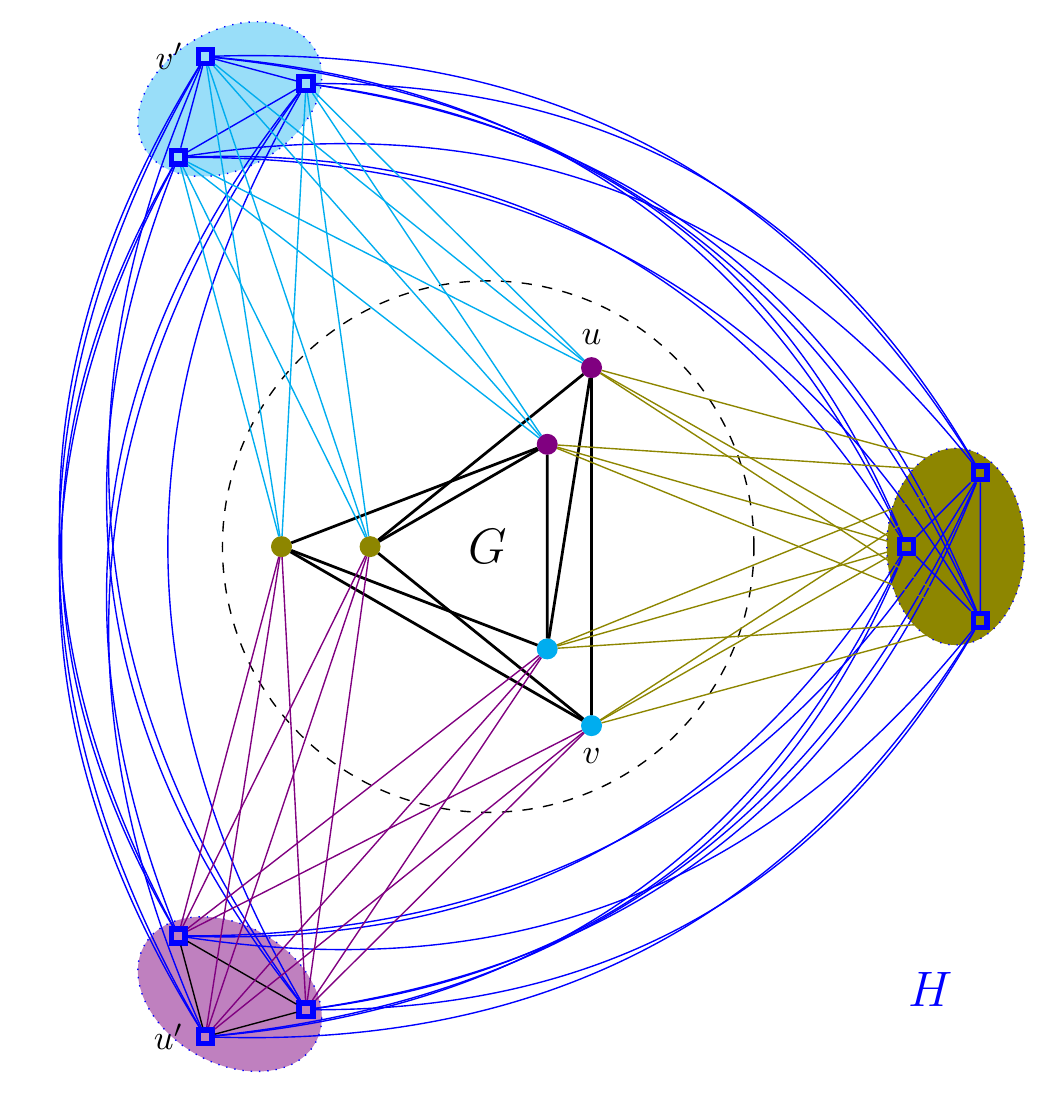} 
\caption{Illustration for Reduction~\ref{reduction:colored}.  Inside
  the dashed circle is a $3$-regular graph $G$, whose six vertices are
  colored with three colors.  Three sets of new vertices (blue
  squares) are added, each for a different color.  For example, the
  set containing $v'$ is for color cyan, and are connected by cyan
  edges to olive and violet vertices of $G$.  Also added are (blue)
  edges among all vertices in $U$.}
  \label{fig:fine-reduction}
\end{figure}

Before proving the main theorems of this paper, let us recall some
simple facts on minimum vertex covers of a {$d$-degree-bounded} graph
$G$ with $d \ge 3$.  Trivially, $\tau(G) < |V(G)|$, while an easy
degree counting tells us that $\tau(G) \ge |V(G)| / (d + 1)$.  If $G$
does contain a clique on $d + 1$ vertices, this clique is necessarily
a component of $G$.  For the vertex cover problem, such a clique
component would not concern us: We take $d$ vertices from it, which is
optimal.  For approximation, if we can get approximation ratio
$\alpha$ for the rest of the graph with all such components removed,
then we can have a ratio for the original graph no worse than
$\alpha$.  In other words, any approximation lower bounds of the
vertex cover problem on {$d$-degree-bounded} graphs hold for
{$d$-degree-bounded} graphs with no ($d + 1$)-cliques.
\begin{lemma}\label{thm:apx-fill-in}
  Let $\epsilon$ be a positive constant and let $d \ge 3$ be an
  integer.  If there is an $f(N, \epsilon)$-time approximation
  algorithm with ratio $1 + \epsilon/3$ for the minimum fill-in
  problem on $N$-vertex graphs, then there is an $O(f(c N, \epsilon) +
  N^2)$-time approximation algorithm with ratio $1 + \epsilon$ for the
  minimum vertex cover problem on $N$-vertex {$d$-degree-bounded} graphs, where
  $c$ is a constant depending on only $d$ and $\epsilon$.
\end{lemma}
\begin{proof}
  Let $\alpha = 1 + \epsilon/3$.  We use the $\alpha$-approximation
  algorithm for the minimum fill-in problem to construct a ($1 +
  \epsilon$)-approximation algorithm for the vertex cover problem on
  {$d$-degree-bounded} graphs as follows.  Since it is easy to find a
  $2$-approximation for the vertex cover problem, it suffices to
  consider $\epsilon < 1$.

  Let $G$ be a {$d$-degree-bounded} graph on $n$ vertices; we may
  assume without loss of generality that $G$ contains no clique on $d
  + 1$ vertices.  We apply Reduction~\ref{reduction:colored} to $G$
  with $b := \lceil {\epsilon}^{-1}\rceil$, and let $H$ be the
  obtained graph.  Then we use the $\alpha$-approximation algorithm to
  find a fill-in $E_+$ of $H$, with $|E_+| \le \alpha\cdot \opt{H}$.
  Our algorithm for vertex cover simply returns the set $C$ of full
  vertices: By Proposition~\ref{lem:vertex-cover}, it is a vertex
  cover of $G$.

  We consider first the approximation ratio of this algorithm.  Since
  $E_+$ contains all the missing edges between $C$ and $U$, it follows
  ${|E_+|} \ge |C| \cdot {b n}$.  On the other hand, $|E_+| \le
  \alpha\cdot \opt{H}$.  Combining them with \eqref{eq:upper-bond}, we
  have
  \[
  |C| \le \frac{|E_+|}{b n} \le \frac{\alpha \cdot \opt{H}}{b n} <
  \alpha \frac{b n \tau(G) + 0.5 \tau^2(G)}{b n} = \alpha \tau(G)
  \Big( 1 + \frac{\tau(G)}{2 b n}\Big).
  \]
  By $b = \lceil {\epsilon}^{-1}\rceil \ge {\epsilon}^{-1}$ and the
  fact $\tau(G) < n$, we can conclude
  \[
  \frac{|C|}{\tau(G)} < \alpha \Big( 1 + \frac{\tau(G)}{2 b n}\Big)
  \le \alpha \big( 1 + \frac{\epsilon}{2} \big) = \big( 1 +
  \frac{\epsilon}{3} \big) \big( 1 + \frac{\epsilon}{2} \big),
  \]
  which is smaller than $1 + \epsilon$ for any $\epsilon < 1$.

  We now calculate the running time of this algorithm.  Note that
  $|V(H)| = b d n + n$.  Let $c = (\epsilon^{-1} + 1) d + 1 > b d +
  1$.  The construction takes $O((cn)^2) = O(n^2)$ time; the
  approximation algorithm for $H$ takes $f(c n, \epsilon)$ time; and
  it takes another $O(n^2)$ time to find and return the vertex cover
  $C$.  Thus, the total running time is $O(f(c n, \epsilon) + n^2)$.
  This concludes the proof.
\end{proof}

\begin{lemma}\label{thm:apx-chordal-completion}
  Let $\epsilon$ be a positive constant and let $d \ge 3$ be an
  integer.  If there is an $f(N, \epsilon)$-time approximation
  algorithm with ratio $1 + \epsilon^2/(10 d^3)$ for the chordal
  completion problem on $N$-vertex graphs, then there is an $O(f(c N,
  \epsilon) + N^2)$-time approximation algorithm with ratio $1 +
  \epsilon$ for the minimum vertex cover problem on $N$-vertex
  {$d$-degree-bounded} graphs, where $c$ is a constant depending on
  only $d$ and $\epsilon$.
\end{lemma}
\begin{proof}
  Let $\alpha = 1 + \epsilon^2/(10 d^3)$.  We use the
  $\alpha$-approximation algorithm for the minimum chordal completion
  problem to construct a ($1 + \epsilon$)-approximation algorithm for
  the vertex cover problem on {$d$-degree-bounded} graphs as follows.
  Since it is easy to find a $2$-approximation for the vertex cover
  problem, it suffices to consider $\epsilon < 1$.

  Let $G$ be a {$d$-degree-bounded} graph on $n$ vertices; we may
  assume without loss of generality that $G$ contains no clique on $d
  + 1$ vertices.  We apply Reduction~\ref{reduction:colored} to $G$
  with $b := \lceil {\epsilon}^{-1}\rceil$, and let $H$ be the
  obtained graph.  Then we use the $\alpha$-approximation algorithm to
  find a chordal completion of $H$; let $\widehat H$ be the obtained
  chordal supergraph.  Recall that the minimum number of edges a
  chordal supergraph of $H$ can have is $|E(H)| + \opt{H}$, thus
  $|E(\widehat H)| \le \alpha \big( |E(H)| + \opt{H} \big)$.  Denote
  by $E_+$ the fill-in produced by this algorithm, i.e., $E_+ =
  E(\widehat H)\setminus E(H)$.  Our algorithm for vertex cover simply
  returns the set $C$ of full vertices: By
  Proposition~\ref{lem:vertex-cover}, it is a vertex cover of $G$.

  We consider first the approximation ratio of this algorithm.  Apart
  from edges of $G$, the constructed graph $H$ contains ${|U| \choose
    2}$ edges among $U$ and $n \cdot (|U| - b n)$ edges between $V(G)$
  and $U$.  Thus,
  \begin{align*}
    |E(H)| &= |E(G)| + {b d n \choose 2} + n (b d n - b n)
    \\
    &\le \frac{d n}{2} + \frac{b^2 d^2 n^2 - b d n}{2} + b (d - 1) n^2
    \\
    &= \frac{d n - b d n}{2} + \frac{b^2 d^2 n^2 + 2 b (d - 1) n^2}{2}
    \\
    &< b^2 d^2 n^2,
  \end{align*}
  where the last inequality follows from that $b \ge \epsilon^{-1} >
  1$ and $d \ge 3$.  Since $E_+$ contains all the missing edges
  between $C$ and $U$, it follows ${|E_+|} \ge |C| \cdot {b n}$.
  Combining \eqref{eq:upper-bond}, we have
  \begin{align*}
    |C| \le \frac{|E_+|}{b n} & = \frac{|E(\widehat H)| - |E(H)|}{b n}
    \\
    & \le \frac{\alpha \big( |E(H)| + \opt{H} \big) - |E(H)|}{b n} &
    \\
    &= \frac{( \alpha - 1) |E(H)| + \alpha \cdot \opt{H} }{b n}
    \\
    &< \frac{( \alpha - 1) b^2 d^2 n^2 + \alpha \big( b n \tau(G) +
      {\tau^2(G)}/2 \big) }{b n}
    \\
    &= ({\alpha - 1}) b d^2 n + \alpha \tau(G) +
    \frac{{\alpha}\tau^2(G) }{2 b n}
    \\
    &< ({\alpha - 1}) b d^2 n + \alpha \tau(G) +
    \frac{{\alpha}\tau(G) }{2 b},
  \end{align*}
  where the last inequality follows from that $ \tau(G) < {n}$.  And
  then
  \begin{align*}
    \frac{|C|}{\tau(G)} &< \frac{({\alpha - 1}) b d^2 n}{\tau(G)} +
    \alpha + \frac{{\alpha} }{2 b}
    \\
    &< 2(\alpha - 1) b d^3 + \alpha + \frac{\alpha}{2 b} & \big(
    \tau(G) > \frac{n}{2 d}\big)
    \\
    &\le 2(\alpha - 1) b d^3 + \alpha + \frac{\alpha \epsilon}{2} &
    \big(b \ge \frac{1}{\epsilon}\big)
    \\
    &= 2(\alpha - 1) b d^3 + (\alpha - 1)\big(1 +
    \frac{\epsilon}{2}\big) + 1 + \frac{\epsilon}{2}
    \\
    &= (\alpha - 1) \big(2 b d^3 + 1 + \frac{\epsilon}{2} \big) + 1 +
    \frac{\epsilon}{2}
    \\
    &< ({\alpha - 1})\big( \frac{4 d^3}{\epsilon} + 1 +
    \frac{\epsilon}{2} \big) + 1 + \frac{\epsilon}{2} & \big(b <
    \frac{2}{\epsilon}\big)
    \\
    &< ({\alpha - 1})\frac{5 d^3}{\epsilon} + 1 + \frac{\epsilon}{2} &
    \big(d \ge 3\; \&\; \epsilon < {1}\big)
    \\
    &= \frac{\epsilon^2}{10 d^3} \frac{5 d^3}{\epsilon} + 1 +
    \frac{\epsilon}{2}
    \\
    &= 1 + {\epsilon}.
  \end{align*}

  We now calculate the running time of this algorithm.  Note that
  $|V(H)| = b d n + n$.  Let $c = (\epsilon^{-1} + 1) d + 1 \ge b d +
  1$.  The construction takes $O((cn)^2) = O(n^2)$ time; the
  approximation algorithm for $H$ takes $f(c n, \epsilon)$ time; and
  it takes another $O(n^2)$ time to find and return the vertex cover
  $C$.  Thus, the total running time is $O(f(c n, \epsilon) + n^2)$.
  This concludes the proof.
\end{proof}

Together with the hardness results of the vertex cover problem on
degree-bounded graphs, it is now quite straightforward to derive the
main results announced in Section~\ref{sec:intro}.
\begin{proof}[Proof of Theorem~\ref{thm:no-ptas}]
  Papadimitriou and Yannakakis~\cite{papadimitriou-91-apx-hardness}
  and Alimonti and Kann~\cite{alimonti-00-apx-hard-cubic-graphs}
  showed that vertex cover is APX-hard on $d$-degree-bounded graphs
  for all $d \ge 3$.  If there is a \ptas{} for the minimum fill-in
  problem or the chordal completion problem, then according to
  Lemma~\ref{thm:apx-fill-in} and
  Lemma~\ref{thm:apx-chordal-completion} respectively, we can use it
  to derive a \ptas{} for the vertex cover problem on
  $3$-degree-bounded graphs, which implies P $=$ NP.
\end{proof}

The following theorem has been essentially observed and used by
Marx~\cite[Lemma~2.5]{marx-07-approximation-optimality}.\footnote{His
  proof, which is omitted in the conference version, uses the
  reduction with bounded-degree expander graphs of Papadimitriou and
  Yannakakis~\cite{papadimitriou-91-apx-hardness} and the
  almost-linear size PCP of Dinur~\cite{dinur-07-pcp} (personal
  communication).} A simpler proof was given by Bonnet et
al.~\cite[Proposition 3, Theorem
9]{bonnet-15-subexponential-inapproximability}.
\begin{lemma}[\cite{marx-07-approximation-optimality,
    bonnet-15-subexponential-inapproximability}]
  \label{thm:bounded-degree-vc}
  Assuming ETH, there exist an integer d and a positive constant
  $\epsilon$ such that there exists no ($1+\epsilon$)-approximation
  algorithm for vertex cover on $d$-degree-bounded graphs in time
  $2^{O(n^{1-\delta})}$ for any positive constant $\delta$.
\end{lemma}

\begin{proof}[Proof of Theorem~\ref{thm:main-2}]
  Suppose that there is a $2^{O(n^{1 - \delta})}$-time algorithm for
  some positive constant $\delta < 1$ that can approximate the minimum
  fill-in problem within ratio $1 + \epsilon$ for any positive
  $\epsilon$.  Then for any integer $d \ge 3$, we can use
  Lemma~\ref{thm:apx-fill-in} to derive a $2^{O(n^{1 - \delta})}$-time
  ($1 + 3\epsilon$)-approximation algorithm for the vertex cover
  problem on $d$-degree-bounded graphs.  Since $\epsilon$ can be made
  arbitrarily small, together with Lemma~\ref{thm:bounded-degree-vc},
  the algorithm refutes \textsc{eth}.  A similar argument works for
  the chordal completion problem.
\end{proof}

\section{Concluding remarks}\label{sec:remarks}
We have shown the first inapproximability result and almost tight
lower bound for parameterized algorithms and exact algorithms for the
minimum fill-in problem.  We also get the first inapproximability
result for the chordal completion problem under the assumption P $\ne$
NP.  All these results are consequences of our new reduction from the
vertex cover problem to the minimum fill-in problem.

It is easy to verify that our reductions work for the completion
problem to any graph class that forbids $C_4$ and contains all split
graphs, e.g., split graphs, $C_4$-free graphs, and even-hole-free
graphs (i.e., $\{C_{2\ell} : \ell \ge 2\}$-free graphs).\footnote{We
  omit the definitions of these graph classes that are nonessential
  for our results and the following discussion.  The reader unfamiliar
  with these notations is referred to
  \url{http://www.graphclasses.org/} for details.}
On the one hand, Proposition~\ref{lem:lower-bound} applies to all
completion to graph classes that forbid $C_4$.  On the other hand, the
solution obtained by connecting $C\cup U$ into a clique is a solution
for any completion problem to a superclass of the class of split
graphs.  In other words, our reductions work for an objective graph
class as long as the set $\cal F$ of its forbidden induced subgraphs
satisfies the following conditions: (1) $C_4\in \cal F$; and (2) any
other $F\in \cal F$ contains $2 K_2$ or $C_5$.

Yannakakis' reduction is far more powerful in this sense.  It directly
applies to interval graphs, unit interval graphs, and strongly chordal
graph, while a slight modification works for trivially perfect graphs
and threshold graphs (see, e.g., Bliznets et
al.~\cite{bliznets-16-lower-bounds} for details).  Interval graphs are
the intersection graphs of intervals on the real line.  The completion
problem to interval graphs is another graph completion problem that
finds important application in sparse matrix computation.
Tarjan~\cite{tarjan-75-graph-theory-and-gaussian-elimination} showed
that it is equivalent to the profile minimization problem.  Many
algorithms have been developed in literature, e.g.,
\cite{ravi-91-approximate-interval-completion,
  villanger-09-interval-completion}.  Since all previous hardness
results use Yannakakis' reduction or some variations, the status of
hardness results on these problems are the same.

We propose here a possible remedy to make our reductions work for
interval graphs.  It needs a stronger claim on coloring subcubic
($3$-degree-bounded) graphs containing no $K_4$---namely, we want to
color it in a way that there is a maximum independent set $X$
receiving only two colors.  If we use such a colored graph in
Reduction~\ref{reduction:colored}, then to get the split graph, we can
also add all edges between the independent set and $V(G)\setminus X$,
which is minimum vertex cover of $G$.  The resulting graph would then
be an interval graph.  However, whether this approach works relies on
two questions that we have no answer: (1) is it true that every cubic
graph containing no $K_4$ admits a $3$-coloring that uses at most two
colors for some maximum independent set of this graph?  and (2) if the
answer to the first question is yes, can such a coloring be found in
polynomial time?  Note that, however, this cannot be generalized to
$d$-degree-bounded graphs, because they would imply a
$2$-approximation algorithm for the maximum independent set problem on
$d$-degree-bounded graphs: by taking the vertices from a largest color
class.

The proof of Yannakakis is very influential for another reason.  He
built the correlation between linear arrangement and minimum fill-in
and related problems, including treewidth (recall that the treewidth
of a graph $G$ can be defined as one less than the size of a maximum
clique of a chordal supergraph of $G$ with the smallest clique
number), interval completion, and pathwidth.  Indeed, Yannakakis
defined the chain graphs to capture the ordering property in a graph
(he has used chain graphs in early work
\cite{yannakakis-81-bipartite-node-deletion} without giving a name).
The linear arrangement problem is a special case of the general family
of graph layout problems, which ask for an ordering of the vertices of
a graph to minimize/maximize some measures.  Similar idea had actually
been used by Kashiwabara and
Fujisawa~\cite{kashiwabara-79-interval-completion,
  kashiwabara-79-pathwidth}.  All hitherto known reductions on related
problems, including \cite{arnborg-87-complexity-fill-in,
  bliznets-16-lower-bounds, wu-14-inapproximability--treewidth},
followed this idea, and used graph layout problems as the source
problems.  An important benefit of this approach is that they can be
(usually in an effortless way) applied to related problems on (proper)
interval graphs.  However, these reductions usually explode the graphs
too much, and their approximation hardness has not been settled.

Another open problem is whether our reduction can be adapted to the
treewidth problem.  Similar as all previous reductions, the graphs
produced by our reductions are very dense.  Therefore, none of them
applies to sparse graphs and planar graphs.  Indeed, the complexity of
treewidth on planar graphs is rather a famous open problem.  As a
final remark, a planar graph has treewidth $O(\sqrt{n})$ and a fill-in
of size $o(n^2)$ \cite{chung-94-planar-fill-in}, while a the treewidth
of a cubic graph can be $\Theta( n)$ and its minimum fill-in can be
$\Theta(n^2)$~\cite{agrawal-93-minimum-fill-in}.

\paragraph{Acknowledgment.}  We are indebted to Marek Cygan and
Bingkai Lin for careful reading of this manuscript and helpful comments.

{
  \small
  \bibliographystyle{plainurl}
  \bibliography{../../journal,../../main}
}
\end{document}